\documentclass[10pt,conference]{IEEEtran}
\usepackage{epsfig,epsf}
\usepackage{color}
\usepackage{bbold}
\usepackage{multirow}
\usepackage{latexsym}
\usepackage{amssymb}
\usepackage{amsmath}
\usepackage{amsfonts}
\usepackage[sort]{cite}
\usepackage{fullpage}
\usepackage{mathrsfs}
\usepackage{scalefnt}
\usepackage{extarrows}
\usepackage{setspace}
\usepackage{stfloats}
\usepackage{url}
\usepackage{hyphenat}
\usepackage{graphicx}
\usepackage[table]{xcolor}
\usepackage{bbm}
\usepackage{dsfont}
\usepackage{bm}
\usepackage{array}
\usepackage{wrapfig}
\usepackage[]{algpseudocode}
\usepackage{algorithm}

\usepackage{mathtools}
\allowdisplaybreaks[1]
\usepackage{epstopdf}
\newcommand\numberthis{\addtocounter{equation}{1}\tag{\theequation}}
\pdfpageattr {/Group << /S /Transparency /I true /CS /DeviceRGB>>}

\makeatletter
\newsavebox\myboxA
\newsavebox\myboxB
\newlength\mylenA
\newcommand*\xoverline[2][0.75]{%
    \sbox{\myboxA}{$\m@th#2$}%
    \setbox\myboxB\null
    \ht\myboxB=\ht\myboxA%
    \dp\myboxB=\dp\myboxA%
    \wd\myboxB=#1\wd\myboxA
    \sbox\myboxB{$\m@th\overline{\copy\myboxB}$}
    \setlength\mylenA{\the\wd\myboxA}
    \addtolength\mylenA{-\the\wd\myboxB}%
    \ifdim\wd\myboxB<\wd\myboxA%
       \rlap{\hskip 0.5\mylenA\usebox\myboxB}{\usebox\myboxA}%
    \else
        \hskip -0.5\mylenA\rlap{\usebox\myboxA}{\hskip 0.5\mylenA\usebox\myboxB}%
    \fi}
\makeatother

\def\phi{\varphi}

\def\cC{{\cal C}}

\def\bff{{\mathbf{f}}}

\def\br{{\mathbf{r}}}

\def\bt{{\mathbf{t}}}

\def\bv{{\mathbf{v}}}

\def\bx{{\mathbf{x}}}
\def\by{{\mathbf{y}}}

\def\b0{{\mathbf{0}}}

\def\argmin{\mathop{\mathrm{argmin}}}
\def\argmax{\mathop{\mathrm{argmax}}}

\def\cA{\mathcal{A}}
\def\cB{\mathcal{B}}
\def\cC{\mathcal{C}}

\def\cI{\mathcal{I}}

\def\cN{\mathcal{N}}

\def\cS{\mathcal{S}}

\def\cX{\mathcal{X}}

\DeclarePairedDelimiter\abs{\lvert}{\rvert}%

\def \Re[#1]{#1_{\rm R}}
\def \Im[#1]{#1_{\rm I}}
\def \Cpx[#1]{{#1_{\rm C}}}

\def \Cpxe[#1]{\tilde{#1}}

\def \tx {x} 
\def \rx {y} 

\def \txv {\bx} 
\def \rxv {\by} 

\def \txvl {\txv_{\rm{L}}}
\def \rxvl {\rxv_{\rm{L}}}

\def \snr {\rho}

\def \E {\mathbb{E}} 
\def \Pr {\mathbb{P}} 

\def \tn {M} 
\def \rn {N} 
\def \nv {\bv} 

\def \sign {\text{sign}}

\def \bfv {\bff} 

\def \arrt {\bt}
\def \arrr {\br}

\def \thetar {\theta_{\rm R}}
\def \thetat {\theta_{\rm T}}

\def\angle{{\operatorname{angle}}}
\def\real{{\operatorname{real}}}
\def\imag{{\operatorname{imag}}}
\def \sL {s_{\rm{L}}}

\def \sA {\cS}
\def \sAL {\cS_{\rm{L}}}
\renewcommand{\angle}{\measuredangle}

\newtheorem{prop}{Property}

\newcommand{\C}{\mathbb{C}}

\begin{document}

\title{Beamforming with Multiple One-Bit Wireless Transceivers
}
\author{\IEEEauthorblockN{Kang Gao, J. Nicholas Laneman, Bertrand Hochwald}
\IEEEauthorblockA{University of Notre Dame\\
Email: \texttt{\{kgao,jnl,bhochwald\}@nd.edu}}}
\maketitle


\pagenumbering{gobble}

\begin{abstract}

Classical beamforming techniques rely on highly linear transmitters and receivers to allow phase-coherent combining at the transmitter and receiver.  The transmitter uses beamforming to steer signal power towards the receiver, and the receiver uses beamforming to gather and coherently combine the signals from multiple receiver antennas.  When the transmitters and receivers are instead constrained for power and cost reasons to be non-linear one-bit devices, the potential advantages and performance metrics associated with beamforming are not as well understood.  We define beamforming at the transmitter as a codebook design problem to maximize the minimum distance between codewords.  We define beamforming at the receiver as the maximum likelihood detector of the transmitted codeword. We show that beamforming with one-bit transceivers is a constellation design problem, and that we can come within a few dB SNR of the capacity attained by linear transceivers.

\end{abstract}
\IEEEpeerreviewmaketitle
\section{Introduction}
Simple one-bit wireless transceivers are being considered for a variety of cost, size and power-related reasons, especially as mobile wireless communications moves to the millimeter-wave band \cite{singh2009limits,krone2012capacity,mo2015capacity,mo2014high,mezghani2009analysis,mezghani2007ultra, mezghani2008analysis,mollen2016one,mollen2017uplink,choi2016near,li2016channel,mo2014channel,studer2015quantized,gao2017power,saxena2016analysis,li2017downlink,jacobsson2017massive}.  Multiple transceiver chains are being considered to allow beamforming at the transmitter and/or receiver to regain signal energy lost to path and penetration losses at such high carrier frequencies.  Yet it is unclear what it means to beamform with one-bit transceivers.

Classical beamforming techniques that require highly linear transmitters and receivers are well-understood.  They are implemented using high resolution analog-to-digital converters (ADCs) and digital-to-analog converters (DACs). Because such ADCs and DACs are power hungry (for example, a 12-bit 4 Gsample/second ADC (Texas Instruments ADC12J4000) consumes two Watts (2 W) \cite{bib:TI}), low-resolution (especially one-bit) ADCs \cite{singh2009limits,krone2012capacity,mo2015capacity,mo2014high,mezghani2009analysis,mezghani2007ultra, mezghani2008analysis,mollen2016one,mollen2017uplink,choi2016near,li2016channel,mo2014channel,studer2015quantized,gao2017power} and DACs \cite{gao2017power,saxena2016analysis,li2017downlink,jacobsson2017massive} are being considered instead.  With such non-linear devices, the beamforming techniques and corresponding performance metrics are not well-understood.  In this paper, provide some simple techniques and performance metrics.

\section{System Model}
We focus on a model where one-bit quantization is considered at both the transmitter and receiver in a line-of-sight (LOS) channel:
\begin{equation}
\rxv=\sign\left(\sqrt{\frac{\snr}{\tn}}\arrr\arrt^H\txv+\nv\right),
\label{eq:complex_channel_model_LOS_channel}
\end{equation}
where $\tn$ and $\rn$ are the number of transmitters and receivers, $\txv\in\{\pm \frac{1}{\sqrt{2}}\pm\frac{j}{\sqrt{2}}\}^{\tn}$ and $\rxv\in\{\pm1\pm j\}^{\rn}$ are the transmitted and received signals, $\arrt $ and $\arrr$ are the array responses of the transmitter and the receiver, which are vectors with $\tn$ and $\rn$ complex elements whose magnitudes are 1, $\snr$ is the received SNR at each receive antenna, $\nv\in\C^{\rn}$ is the additive complex Gaussian noise with $\nv\sim\cC\cN(0,I)$ and $\nv$ is independent of $\txv, \arrt,$ and $\arrr$.
The function $\sign(\cdot)$ provides the sign of the real and imaginary part of the input as the real and imaginary part of the output.

A quick observation is that we can combine the array response at the transmitter $\arrt $ and the transmitted vector $\txv$ in (\ref{eq:complex_channel_model_LOS_channel}) and get an equivalent single-input multiple-output (SIMO) model:
\begin{equation}
\rxv=\sign\left(\sqrt{\snr}\arrr s+\nv\right),
\label{eq:equivalent_SIMO_model}
\end{equation}
with
\begin{equation}
s=\sqrt{\frac{1}{\tn}}\arrt^H\txv,\quad\txv\in\{\pm\frac{1}{\sqrt{2}}\pm \frac{j}{\sqrt{2}}\}^{\tn},
\label{eq:equivalent_transmitted_signal}
\end{equation}
and $s$ can be considered as the equivalent transmitted symbol in the equivalent SIMO model.

For comparison, we also show the equivalent SIMO model for a linear system:
\begin{equation}
\rxvl=\sqrt{\snr}\arrr\sL+\nv,
\label{eq:equivalent_SIMO_model_linear}
\end{equation}
with
\begin{equation}
\sL=\sqrt{\frac{1}{\tn}}\arrt^H\txvl,\quad \txvl^H\txvl=\tn,
\label{eq:equivalent_transmitted_signal_linear}
\end{equation}
where $\txvl$ is the linear transmitted vector with total power $\tn$, $\sL$ is the equivalent transmitted symbol, $\rxvl$ is the linear received vector, and $\nv\sim\cC\cN(0,I)$ is complex additive Gaussian noise. 

\section{Problem Description}
In classical beamforming techniques, the transmitter steers signal power to the receiver to maximize the distance between transmitted symbols, and the receiver combines the received signal coherently to effectively boost the received SNR and reduce the probability of error in detection. A similar idea can be applied to set up the problem of beamforming with multiple one-bit transceivers, and we compare linear and one-bit beamforming throughout.


\subsection{Beamforming at the transmitter}
Transmitter beamforming can be expressed as a codebook design problem.

\subsubsection{classical linear transceivers}
In a classical system, the codebook design problem is
\begin{equation}
    \cB_{\rm{L}}=\argmax_{\substack{\cX\subset \{\txv:\txv^H\txv=\tn\}\\|\cX|=K}}\min_{\substack{\txv_i,\txv_j\in\cX\\ \txv_i\neq \txv_j}} |s_i-s_j|, 
    \label{eq:linear_bf_tx_problem_def}
\end{equation}
where $s_i = \sqrt{\frac{1}{\tn}}\arrt^H\txv_i$ is the corresponding equivalent transmitted symbol of vector $\txv_i$, $K$ is the number of vectors in the codebook.

This problem is equivalent of finding a set of symbols with size $K$ : 
\begin{equation}
    \sAL=\argmax_{\substack{\cC\subset \cA_{\rm{L}},|\cC|=K}}\min_{\substack{s_i,s_j\in\cC\\ s_i\neq s_j}} |s_i-s_j|, 
    \label{eq:linear_bf_tx_problem_def_symbol}
\end{equation}
where
\begin{equation}
    \cA_{\rm{L}}=\{\sL:\sL=\sqrt{\frac{1}{\tn}}\arrt^H\txvl,\txvl^H\txvl=\tn\},
\end{equation}
and then find the corresponding vectors $\txv$ of those symbols.  This is classically solved (approximately) by setting $\txvl=\arrt u$, where $u$ is a symbol generally taken from a standard PSK or QAM constellation.

\subsubsection{one-bit transceivers}
With one bit quantization, similar to (\ref{eq:linear_bf_tx_problem_def}), the design problem is
\begin{equation}
    \cB = \argmax_{\substack{\cX\subset \{\pm\frac{1}{\sqrt{2}}\pm\frac{j}{\sqrt{2}}\}^\tn\\|\cX|=K}}\min_{\substack{\txv_i,\txv_j\in\cX\\ \txv_i\neq \txv_j}} |s_i-s_j|,
    \label{eq:one_bit_bf_tx_problem_def}
\end{equation}
where $s_i = \sqrt{\frac{1}{\tn}}\arrt^H\txv_i$, which is the corresponding equivalent transmitted symbol of vector $\txv_i$.

Similarly, the corresponding equivalent symbol design problem is
\begin{equation}
    \sA=\argmax_{\substack{\cC\subset \cA,|\cC|=K}}\min_{\substack{s_i,s_j\in\cC\\ s_i\neq s_j}} |s_i-s_j|, 
    \label{eq:one_bit_bf_tx_problem_def_symbol}
\end{equation}
where
\begin{equation}
    \cA=\{s:s=\sqrt{\frac{1}{\tn}}\arrt^H\txv,\quad\txv\in\{\pm\frac{1}{\sqrt{2}}\pm \frac{j}{\sqrt{2}}\}^{\tn}\}.
    \label{eq:all_one_bit_symbol}
\end{equation}
There is no equivalent classical solution to this problem, and we discuss some approximate solutions.

\subsection{Beamforming at the receiver}
Beamforming at the receiver minimizes the probability of error in the detection of $s_{\rm{L}} \in \sAL$ or $s\in \sA$.  We consider the maximum-likelihood (ML) detector, which minimizes the probability of error when the input is uniformly distributed.

\subsubsection{classical linear transceivers}
The ML decoder of classical linear transceivers is
\begin{equation}
    \hat{s}_{\rm{L}}=\argmax_{s_{\rm{L}}\in\sAL}\Pr(\rxvl|s_{\rm{L}}).
    \label{eq:prob_def_ML_linear}
\end{equation}
\subsubsection{one-bit transceivers}
The ML decoder of one-bit transceivers is
\begin{equation}
    \hat{s}=\argmax_{s\in\sA}\Pr(\rxv|s).
    \label{eq:prob_def_ML_one_bit}
\end{equation}

\section{Beamforming at the transmitter}
\label{sec:BF_tx}
\subsection{classical linear transceivers}
\label{sec:BF_tx_linear}
The design of codebook $\cB_{\rm{L}}$ shown in  (\ref{eq:linear_bf_tx_problem_def}) is related to the problem of circle packing in a circle \cite{wiki:circle_in_circ} which is an open problem in general. However, if we restrict the symbols $s_{\rm{L}}$ (\ref{eq:equivalent_transmitted_signal_linear}) obtained from $\txvl\in\cB_{\rm{L}}$ to have the largest magnitude, we obtain an approximated solution
\begin{equation}
    {\cB}_{\rm{L}}=\{\arrt e^{j\frac{2\pi n}{K}}, n=0,1,\cdots, K-1\},
\end{equation} 
which is optimum when $K\leq 6$ according to \cite{graham1998dense}. The corresponding alphabet of equivalent transmitted symbols becomes
\begin{equation}
    \cS_{\rm{L}}=\{\sqrt{\tn}e^{j\frac{2\pi n}{K}}, n=0,1,\cdots, K-1\},
\end{equation}
which is a K-PSK modulation with magnitude $\sqrt{\tn}$.

\subsection{one-bit transceivers}
The design of the codebook $\cB$ shown in (\ref{eq:one_bit_bf_tx_problem_def}) requires searching among $4^\tn$ symbols. Rather than do the complete search, we search over a much smaller subset. When $K=2$ and $K=4$, the solution of (\ref{eq:one_bit_bf_tx_problem_def_symbol}) is
\begin{equation}
    \cS=\{\pm s_{\rm{max}}\}, \quad K=2,
    \label{eq:K_2_cS}
\end{equation}
\begin{equation}
    \cS=\{\pm s_{\rm{max}},\pm j s_{\rm{max}}\}, \quad K=4,
    \label{eq:K_4_cS}
\end{equation}
where
\begin{equation}
s_{\rm{max}}=\argmax_{s\in\cA}|s|.
\label{eq:max_mag_s}
\end{equation}
Finding the $\txv$ that corresponds to $s_{\rm{max}}$ can be done by searching all $4^M$ possible vectors to find $\txv_{\rm max}$.  We suggest a simpler method.

We first define a subset of $\{\pm\frac{1}{\sqrt{2}}\pm \frac{j}{\sqrt{2}}\}^{\tn}$ to be
\begin{align*}
    \tilde{\cX}=\{&\txv(\phi):\phi\in[0,2\pi],\real(t_ke^{j\phi})\neq 0,\\ 
    &\imag(t_ke^{j\phi})\neq 0, 1\leq k\leq \tn\},
    \numberthis
    \label{eq:tilde_X_def}
\end{align*}
where $t_k$ is the $k$th element of $\arrt$, and $\txv(\phi)$ is defined as:
\begin{equation}
    \txv(\phi) = \frac{1}{\sqrt{2}}\sign(\arrt e^{j\phi}).
    \label{eq:tx_beamforming_transmitted_vector}
\end{equation}

The corresponding set of equivalent transmitted symbols is defined as
\begin{equation}
    \cS(\tilde{\cX})=\{s:s=\sqrt{\frac{1}{\tn}}\arrt^H\txv,\txv\in\tilde{\cX}\}.
    \label{eq:subset_tx_equiv_symbol}
\end{equation}

Even though we have infinitely many $\phi$ in the interval $[0,2\pi]$, the size of $\tilde{\cX}$ is bounded by $4M$.
Let $\tx_k(\phi)$ be the $k$th element of $\txv(\phi)$.
By varying $\phi$ from $0$ to $2\pi$, the value of $\tx_k(\phi)$ potentially changes 4 times for each $k$, and therefore we will get at most $4\tn$ different $\txv$.

For any complex number $c\in\C$, we have
\begin{equation}
    \angle{(c^*\sign(c))}\in(-\frac{\pi}{4},\frac{\pi}{4}).
\end{equation} 
Therefore,
\begin{align*}
    \angle{(t_k^*\tx_k(\phi))}&=\angle{(e^{j\phi}(t_ke^{j\phi})^*\sign(t_ke^{j\phi}))}\\
    &\in\left(\phi-\frac{\pi}{4},\phi+\frac{\pi}{4}\right),
    \numberthis
    \label{eq:angle_align_prop}
\end{align*}
for any $k$. Also, there is no other vector $\txv\in\{\pm\frac{1}{\sqrt{2}}\pm\frac{j}{\sqrt{2}}\}^{\tn}$ that satisfies $t_k^*\tx_k\in(\phi-\frac{\pi}{4},\phi+\frac{\pi}{4})$ for all $1\leq k \leq \tn$.
We have some properties of the set $\cS(\tilde{\cX})$.

\begin{prop}
For any $s\in\cS(\tilde{\cX})$, we have $|s|> \sqrt{\frac{\tn}{2}}$.
\end{prop}
This result gives a lower bound on the ``beamforming gain" that can be expected with one-bit transceivers.

\begin{proof}
According to (\ref{eq:subset_tx_equiv_symbol}), for any $s\in\cS(\tilde{\cX})$, there $\exists \quad \phi\in[0,2\pi]$, so that
\begin{equation}
    s=\sqrt{\frac{1}{\tn}}\arrt^H\txv(\phi).
\end{equation}
Let $c_k=t_k^*\tx(\phi)_k$, we have
\begin{equation}
    |c_k|=1,\angle{(c_k)}\in(\phi-\frac{\pi}{4},\phi+\frac{\pi}{4}).
\end{equation}
Therefore,
\begin{align*}
    |s|&=\sqrt{\frac{1}{\tn}}\abs{\sum_{k=1}^{\tn}c_k}=\sqrt{\frac{1}{\tn}}\abs{\sum_{k=1}^{\tn}c_k e^{-j\phi}}\\
    &\geq \sqrt{\frac{1}{\tn}}\sum_{k=1}^{\tn}\real{\big(c_k e^{-j\phi}\big)} > \sqrt{\frac{\tn}{2}}.
\end{align*}
\end{proof}

\begin{prop}
$s_{\rm{max}}\in\cS(\tilde{\cX})$ and $|s_{\rm{max}}|\geq\frac{2\sqrt{2\tn}}{\pi}$, where $s_{\rm{max}}$ is defined in (\ref{eq:max_mag_s}).
\end{prop}

\begin{proof}
Let 
\begin{equation}
\txv_{\rm{max}}=\argmax_{\txv\in\{\pm\frac{1}{\sqrt{2}}\pm \frac{j}{\sqrt{2}}\}^{\tn}} |\arrt^H\txv|,
\end{equation}
we have $s_{\rm{max}}=\sqrt{\frac{1}{\tn}}\arrt^H\txv_{\rm{max}}$. We will first show $\txv_{\rm{max}}\in\tilde{\cX}$ which indicates $s_{\rm{max}}\in\cS(\tilde{\cX})$.

$\forall \txv\notin \tilde{\cX}$, let $A=\arrt^H\txv,\phi_A=\angle{A},c_k=t^*_k\tx_k$. Then, there exists some $n$ so that $\angle{c_n}\notin(\phi_A-\frac{\pi}{4},\phi+\frac{\pi}{4})$. Otherwise, we have $\txv=\txv(\phi_A)\in\tilde{\cX}$.

We replace the $n$th element of $\txv$ with
\begin{equation}
    \tilde{\tx}_{n}=\frac{1}{\sqrt{2}}\sign(t_ne^{j\phi_A}),
    \label{eq:bf_design_element}
\end{equation}
and denote the new vector as $\tilde{\txv}$.

Let $\tilde{c}_n=(t_n^*\tilde{\tx}_n)$, we have $\angle{\tilde{c}_n}\in(\phi_A-\frac{\pi}{4},\phi_A+\frac{\pi}{4})$.

Also,
\begin{align*}
    &|\arrt^H\tilde{\txv}|^2=|A-c_n+\tilde{c}_{n}|^2\\
    =&|A|^2+2\real(A^*\tilde{c}_n)-2\real(A^*c_n)+(2-2\real(c_n^*\tilde{c}_n))\\
    >& |A|^2 + 2|A|\cos(\phi_A-\angle(\tilde{c}_n)) - 2|A|\cos(\phi_A-\angle(c_n))\\
    >&|A|^2.
\end{align*}
Therefore, $|\arrt^H\tilde{\txv}|>|\arrt^H\txv|$, which means $\txv\neq\txv_{\rm{max}}$.

Hence $\txv_{\rm{max}}\in\tilde{\cX}$ and therefore  $s_{\rm{max}}\in\cS(\tilde{\cX})$.

Now, we will prove $|s_{\rm{max}}|\geq\frac{2\sqrt{2\tn}}{\pi}$.
Since $s_{\rm{max}}\in\cS(\tilde{\cX})$, we have
\begin{align*}
    |s_{\rm{max}}|=\sqrt{\frac{1}{\tn}}\max_{\phi}\abs{\arrt^H\txv(\phi)}\\
    \geq \sqrt{\frac{1}{\tn}}\frac{2}{\pi}\int_{-\frac{\pi}{4}}^{\frac{\pi}{4}}\abs{\arrt^H\txv(\phi)}d \phi.
    \numberthis
    \label{eq:max_tx_bf_bound}
\end{align*}

Let $c_k(\phi)=t_k^*\tx_k(\phi)$ and we have
\begin{equation}
    \abs{\arrt^H\txv(\phi)}=\abs{\sum_{k=1}^{\tn}c_k(\phi)e^{-j\phi}}\geq \sum_{k=1}^{\tn}\cos(\beta_k(\phi))
\end{equation}
with 
\begin{equation}
    \beta_k(\phi) = \angle{c_k(\phi)}-\phi\in[-\frac{\pi}{4},\frac{\pi}{4}].
\end{equation}

When $\phi$ covers $[-\frac{\pi}{4},\frac{\pi}{4}]$, $\beta_k(\phi)$ will also cover $[-\frac{\pi}{4},\frac{\pi}{4}]$ for any $1\leq k\leq \tn$. Therefore, (\ref{eq:max_tx_bf_bound}) becomes
\begin{equation}
    |s_{\rm{max}}|\geq \sqrt{\frac{1}{\tn}}\frac{2}{\pi}\sum_{k=1}^{\tn}\int_{-\frac{\pi}{4}}^{\frac{\pi}{4}}\cos(\beta_k)d \beta_k = \frac{2\sqrt{2\tn}}{\pi}.
\end{equation}

\end{proof}

We use the set $\cS(\tilde{\cX})$ whose size is no larger than $4\tn$ to find $K$ symbols to maximize the minimum distance. The set of the symbols selected for transmission is
\begin{equation}
    {\sA}=\argmax_{\substack{\cC\subset \sA(\tilde{\cX})\\|\cC|=K}}\min_{\substack{s_i,s_j\in\cC\\ s_i\neq s_j}} |s_i-s_j|,
    \label{eq:one_bit_bf_tx_subset_search_symbol}
\end{equation}
which is an approximate solution of (\ref{eq:one_bit_bf_tx_problem_def_symbol}).
Specially, when $K=2,4$, we can quickly find $s_{\rm{max}}$ and apply (\ref{eq:K_2_cS}) and (\ref{eq:K_4_cS}) to find $\cS$ defined in (\ref{eq:one_bit_bf_tx_problem_def_symbol}).  Searching over $4M$ possible $\txv$ is clearly much easier than searching over all $4^M$ possible.

Here, we provide an algorithm to quickly obtain $\tilde{\cX}$ and then ${\cS(\tilde{\cX})}$ can be computed through (\ref{eq:subset_tx_equiv_symbol}) directly.  The algorithm computes $\txv\in\tilde{\cX}$ by varying $\phi$ from $-\angle{t_k}-\epsilon$ to $-\angle{t_k}+\epsilon$, where $\epsilon$ is a small positive value.  By symmetry, we only need to locate $\tn$ possible $\phi$ where an element of $\txv(\phi)$ changes to obtain all the vectors of $\tilde{\cX}$. Algorithm \ref{alg:BF_tx} below does the job:

\begin{algorithm}
	\caption{Codebook design}\label{alg:BF_tx}
    \scriptsize
        \begin{algorithmic}
            \Require {$\tn,\arrt$}
            \State {$\#\tn$ the number of transmitters}
            \State {$\#\arrt$ array response at the transmitter}
        	\State {$\epsilon=10^{-6}$;}
        	\For{$k=1,\cdots,\tn$}
        	    \State{$\phi_k=-\angle t_k+\epsilon$;}
        	   \State{$\txv_k=\frac{1}{\sqrt{2}}\sign(\arrt e^{j\phi_k})$;}
        	\EndFor
        	\begin{equation*}
        	    \tilde{\cX}=\{\pm \txv_k,\pm j\txv_k: k=1,2,\cdots,\tn\}
        	\end{equation*}
        	\Ensure {$\tilde{\cX}$}
        \end{algorithmic}
\end{algorithm}

\section{beamforming at the receiver}
\subsubsection{classical linear transceivers}
For a classical linear system, we consider the equivalent SIMO model shown in (\ref{eq:equivalent_SIMO_model_linear}) and maximal-ratio combining (MRC) is applied to solve (\ref{eq:prob_def_ML_linear})
\begin{equation}
    \bfv_{\rm{L}} = \frac{\arrr^H}{\sqrt{\snr}\arrr^H\arrr}
\end{equation}
\begin{equation}
    \hat{s}_{\rm{L}} = \argmin_{\sL\in\sAL}|\bfv_{\rm{L}}^H\rxvl - \sL|
    \label{eq:beamforming_at_rx_detector}
\end{equation}
where $\sAL$ is the alphabet of the equivalent transmitted symbol $\sL$, $\bfv_{\rm{L}}$ is the linear combining beamforming vector, $\hat{s}_{\rm{L}}$ is the estimate of $\sL$. 
\subsubsection{one-bit transceivers}
For one-bit transceivers, based on the equivalent SIMO model shown in (\ref{eq:equivalent_SIMO_model}), the ML decoder in (\ref{eq:prob_def_ML_one_bit}) can be written as
\begin{equation}
    \hat{s}=\argmax_{s\in\sA}\sum_{k=1}^{\rn}\log \Pr(\rx_k|s),
    \label{eq:ML_detector_one_bit}
\end{equation}
where $\sA$ is the alphabet of the equivalent transmitted symbol $s$, $\rx_k\in\{\pm1\pm j\}$ is the $k$th element of $\rxv$.

Let $\rx_{\rm{R},k}$ and $\rx_{\rm{I},k}$ to be the real part and imaginary part of $\rx_k$. Then, according to model (\ref{eq:equivalent_SIMO_model}), we have
\begin{equation}
    \log \Pr(\rx_k|s) = \log \Pr(\rx_{\rm{R},k}|s) + \log \Pr(\rx_{\rm{I},k}|s),
\end{equation}
with
\begin{equation}
    \Pr(\rx_{\rm{R},k}|s) = Q(-\sqrt{2\snr}\rx_{\rm{R},k}\cdot\real{(r_k s)})
    \label{eq:real_y}
\end{equation}
\begin{equation}
    \Pr(\rx_{\rm{I},k}|s) = Q(-\sqrt{2\snr}\rx_{\rm{I},k}\cdot\imag{(r_k s)}),
    \label{eq:imag_y}
\end{equation}
where $Q(\cdot)$ is the classical Q-function, $\real(\cdot)$ and $\imag(\cdot)$ output the real and imaginary part of a complex number, $r_k$ is the $k$th element of $\arrr$, which is the array response at the receiver.

Let 
\begin{equation}
    p_{\rm{B},k}(s)=\Pr(\rx_{\rm{B},k}=1|s),q_{\rm{B},k}(s)=\Pr(\rx_{\rm{B},k}=-1|s)
    \label{eq:p_q_in_beamforming}
\end{equation}
with $\rm{B}\in\{\rm{R,I}\}$. Then we have
\begin{equation}
   \Pr(\rx_{\rm{B},k}|s)= (p_{\rm{B},k}(s))^\frac{1+\rx_{\rm{B},k}}{2}(q_{\rm{B},k}(s))^\frac{1-\rx_{\rm{B},k}}{2}.
\end{equation}

Therefore,
\begin{align*}
   &\log\Pr(\rx_{\rm{B},k}|s)\\
   &= \frac{1+\rx_{\rm{B},k}}{2}\log p_{\rm{B},k}(s) + \frac{1-\rx_{\rm{B},k}}{2}\log q_{\rm{B},k}(s)\\
   &=\frac{1}{2}\left(\rx_{\rm{B},k}\log\frac{p_{\rm{B},k}(s)}{q_{\rm{B},k}(s)} + \log(p_{\rm{B},k}(s)q_{\rm{B},k}(s))\right). 
\end{align*}

Since
\begin{equation}
    \log \Pr(\rxv|s) = \sum_{k=1}^{\rn}\sum_{\rm{B}\in\{\rm{R},\rm{I}\}}\log \Pr(\rx_{\rm{B},k}|s),
\end{equation}
we have
\begin{equation}
    \log \Pr(\rxv|s) = \frac{1}{2}\left(\real(\bfv^H(s)\rxv)+d(s)\right),
    \label{eq:beamforming_at_rx_log_likelihood}
\end{equation}
where the $k$th element of the beamforming vector $\bfv(s)$
\begin{equation}
    [\bfv(s)]_k = \log\frac{p_{\rm{R},k}(s)}{q_{\rm{R},k}(s)} + j\log\frac{p_{\rm{I},k}(s)}{q_{\rm{I},k}(s)},
\end{equation}
and the offset $d(s)$ is defined as
\begin{equation}
    d(s)=\sum_{k=1}^{\rn}\sum_{\rm{B}\in\{\rm{R},\rm{I}\}}\log(p_{\rm{B},k}(s)q_{\rm{B},k}(s)),
\end{equation}
where $p_{\rm{B},k}(s)$ and $q_{\rm{B},k}(s)$ are defined in (\ref{eq:p_q_in_beamforming}).

Then, the ML detector in (\ref{eq:ML_detector_one_bit}) becomes
\begin{equation}
    \hat{s}=\argmax_{s\in\sA}\left(\real{(\bfv^H(s)\rxv)}+d(s)\right).
\end{equation}
The log-likelihood function of the transmitted symbol $s$ can be computed through (\ref{eq:beamforming_at_rx_log_likelihood}), which provides soft information for decoding an outer channel code.

\section{Example with Uniform Linear Arrays}
\label{sec:numerical_result}
We consider a system using 
uniform linear arrays (ULA) with adjacent distance $\frac{\lambda}{2}$ at both the transmitter and the receiver, where $\lambda$ is the wavelength of the carrier. 
According to \cite{constantine2005antenna}, the array response becomes
\begin{align*}
\arrt &=[1,e^{j\pi\sin\thetat},e^{j2\pi\sin\thetat},\cdots,e^{j(\tn-1)\pi\sin\thetat}]^T.\\
\arrr&=[1,e^{j\pi\sin\thetar},e^{j2\pi\sin\thetar},\cdots,e^{j(\rn-1)\pi\sin\thetar}]^T,
\end{align*}
where $\thetat$ and $\thetar$ are the angle of departure (AoD) and the angle of arrival (AoA). We assume $\thetat=10^\circ$ and $\thetar=10^\circ$ in our examples and consider $\tn=\rn=8$ first and then $\tn=\rn=40$ operating at a lower SNR.

\subsection{$\tn=\rn=8$}
We first consider $\tn=\rn=8$. According to (\ref{eq:all_one_bit_symbol}), we have $4^{\tn}=65536$ possible symbols $s\in\cA$ to choose from. Using Algorithm \ref{alg:BF_tx}, we can get $\tilde{\cX}$ and obtain $\cS(\tilde{\cX})$ from (\ref{eq:subset_tx_equiv_symbol}), which has no more than $4\tn=32$ symbols. The scatter plot of all $s\in\cA$ are shown in a complex plane  in Fig. \ref{fig:QPSK_beamforming_pattern}, where 32 red dots represent the symbols $s\in\cS(\tilde{\cX})$, while all the other possible symbols are in green. Their magnitudes are very close to $\sqrt{\tn}$, which is the maximum magnitude of equivalent transmitted symbols of linear transceivers.

\begin{figure}
\centering
    \includegraphics[width=3.5in]{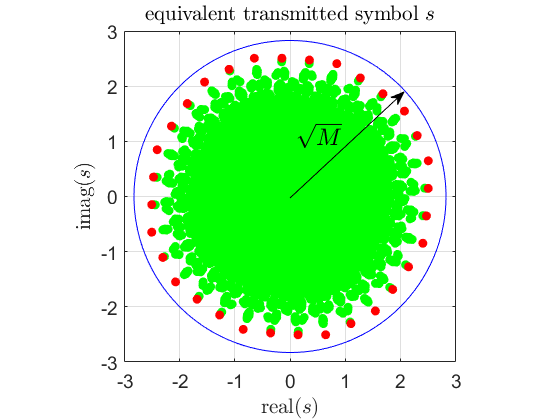}
    \caption{Scatter plot of all $4^\tn=65536$ equivalent transmitted symbols $s\in\cA$ in (\ref{eq:all_one_bit_symbol}) when $\thetat=10^\circ$ and $\tn=8$. The $4M=32$ red dots represent the equivalent transmitted symbols $s\in\cS(\tilde{\cX})$ in (\ref{eq:subset_tx_equiv_symbol}) with $\tilde{\cX}$ obtained through Algorithm \ref{alg:BF_tx}. The blue circle has radius $\sqrt{\tn}$, which is the maximum achievable magnitude of the equivalent transmitted scalar in a linear system with a PSK constellation.}
    \label{fig:QPSK_beamforming_pattern}
\end{figure}

Based on $\tilde{\cX}$ obtained from Algorithm \ref{alg:BF_tx}, we consider $K=2,4,8$ and solve (\ref{eq:one_bit_bf_tx_subset_search_symbol}) for the set of the selected symbols ${\cS}$.  So that the receiver can decode the symbols in ${\cS}$ without knowing the transmitter codebook (which depends on $\arrt$), we desire that the resulting constellation have a regular pre-agreed upon PSK structure.  Fig.\ref{fig:constellation} (red dots) shows the result of choosing QPSK $(K=4)$ and 8-PSK $(K=8)$.  They appear ``rotated", but any such rotation can easily by absorbed into the channel.

\begin{figure}
\centering
    \includegraphics[width=3.4in]{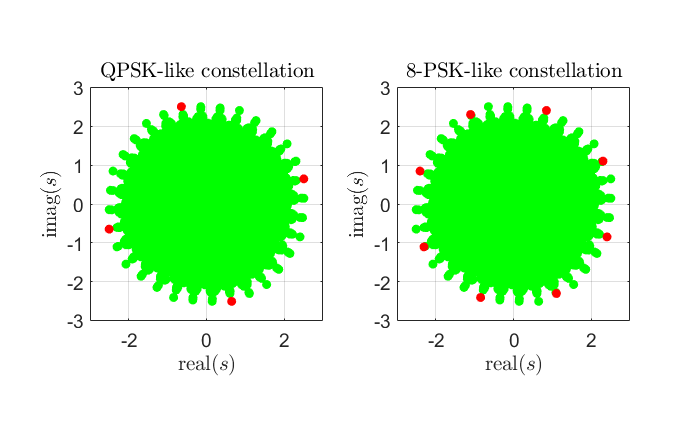}
    \caption{We solve (\ref{eq:one_bit_bf_tx_subset_search_symbol}) to get ${\cS}$ for $K=4$ and $K=8$ (red dots).  The other available symbols $s\in\cA$ (\ref{eq:all_one_bit_symbol}) are in green, and are not used.  For $K=4$, we obtain a QPSK-like constellation, and $K=8$ gives us an 8-PSK-like constellation.}
    \label{fig:constellation}
\end{figure}

We are also able to obtain the mutual information between the input $s$ and the output $\rxv$ when $s$ is uniform input among the BPSK-like (K=2), QPSK-like (K=4), or 8-PSK-like (K=8) constellations.  We have
\begin{equation}
\cI(s\in{\cS};\rxv)=  \E_{s}\left[\sum_{\rxv}\Pr(\rxv|s)\log_2\frac{\Pr(\rxv|s)}{\E_{s}[\Pr(\rxv|s)]}
    \right]
    \label{eq:mutual_info_PSK_like}
\end{equation}
with $s$ uniform distributed among ${\cS}$, and $\Pr(\rxv|s)$ can be easily obtained from the model (\ref{eq:equivalent_SIMO_model}).

Also, we can compute the channel capacity of the system modeled in (\ref{eq:equivalent_SIMO_model}) and (\ref{eq:equivalent_transmitted_signal}), which is equivalent to a discrete memoryless channel (DMC) with $4^\tn$ input and $4^\rn$ output, using Blahut-Arimoto algorithm \cite{arimoto1972algorithm,blahut1972computation}.
This is compared with the channel capacity of a system with linear transceivers, as modeled in (\ref{eq:equivalent_SIMO_model_linear}) and (\ref{eq:equivalent_transmitted_signal_linear}), which is
\begin{equation}
    C_L = \log_2(1+\tn\rn\snr).
    \label{eq:linear_capacity}
\end{equation}

The results are shown in Fig. \ref{fig:QPSK_beamforming_vs_capacity_low_SNR}. We can see that the gap of the channel capacity between the linear transceivers and one-bit transceivers is smaller than 4 dB when the SNR (per receive antenna) is smaller than -10 dB.  We also observe that the BPSK, QPSK, and 8-PSK-like constellations do well at low SNR.
\begin{figure}
\centering
    \includegraphics[width=3.5in]{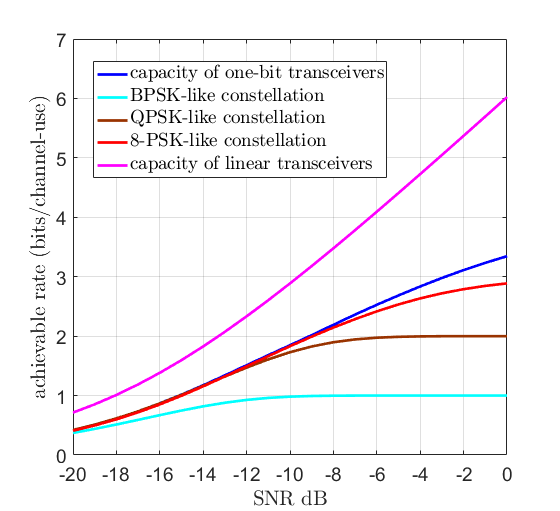}
    \caption{Comparison of the achievable rates using equivalent transmitted symbols $s\in{\cS}$ with $K=2,4,8$, and the channel capacity of systems with one-bit transceivers and linear transceivers for $M=N=8$. When $K=2,4,8$, we have BPSK-like, QPSK-like, and 8-PSK-like constellations, and the achievable rate is computed through (\ref{eq:mutual_info_PSK_like}), which are shown in light blue, brown, and red. The channel capacity of the one-bit transceivers shown in blue is obtained through the Blahut-Arimoto algorithm, and the channel capacity of linear transceivers shown in pink is obtained from (\ref{eq:linear_capacity}).}
    \label{fig:QPSK_beamforming_vs_capacity_low_SNR}
\end{figure}

We now apply an LDPC code, and use receiver beamforming (maximum likelihood) to examine performance. We use a DVB-S.2 standard LDPC code with block size 64800 and code rate 0.5. We employ bit-interleaved coded modulation (BICM \cite{caire1998bit}) with our 8-PSK-like constellation shown in Fig. \ref{fig:constellation}, where the bits generated by the encoder are interleaved before being mapping to the constellation symbols. Gray codes are used to map 3 bits to those 8 symbols.
With 3 bits/symbol and 0.5 code rate, the information rate becomes 1.5 bits/channel-use. 
The log-likelihood of each symbol $s\in\sA$ can be computed using (\ref{eq:beamforming_at_rx_log_likelihood}).

The performance is shown in Fig. \ref{fig:M_8_N_8_beamforming_LDPC}, and we observe that we are only 1.3 dB away from the channel capacity of the one-bit transceivers, and only 4.7 dB away from the channel capacity of the linear transceivers.

\begin{figure}
\centering
    \includegraphics[width=3.2in]{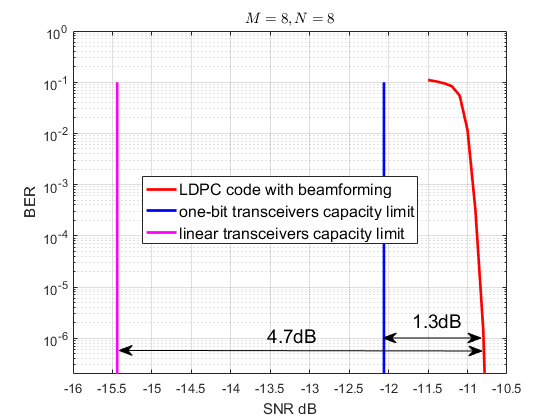}
    \caption{Comparison of the performance of the LDPC code with beamforming for one-bit ULA transceivers, and the channel capacities of one-bit transceivers and linear transceivers.  The information bit rate is 1.5 bits/channel-use with $\tn=\rn=8$ in a LOS channel with $\thetat=\thetar=10^\circ$. The pink and blue vertical lines show the corresponding SNR when the capacity achieve 1.5 bits/channel-use, which can be found from Fig.~\ref{fig:QPSK_beamforming_vs_capacity_low_SNR}. The red curve shows the bit error rate of the LDPC code with code rate 0.5 and 8-PSK-like constellation for transmission.  The gap between the capacity of linear transceivers and the beamformed one-bit transceivers is only 4.7 dB.}
    \label{fig:M_8_N_8_beamforming_LDPC}
\end{figure}

\subsection{$\tn=\rn=40$}
Since the size of $\cS(\tilde{\cX})$ increases linearly with $\tn$, and the complexity of beamforming at the receiver (ML decoder) increases linearly with $\rn$, we may also consider large $\tn$ and $\rn$. 
For example, we consider $\tn=\rn=40$ with the same LOS channel with $\thetat=\thetar=10^\circ$. We again consider $K=8$ and use Algorithm \ref{alg:BF_tx} to obtain the codewords, and seek an information rate of 1.5 bits/channel-use.  The performance is shown in Fig. \ref{fig:M_40_N_40_beamforming_LDPC}, and we are only 4.4 dB from the channel capacity of linear transceivers.  Note the low per-receiver SNR that can be accommodated.  With one-bit transceivers, we are approximately obtaining the $MN$ beamforming gain that is typically obtained with classical linear transceivers.
\begin{figure}
\centering
    \includegraphics[width=3.2in]{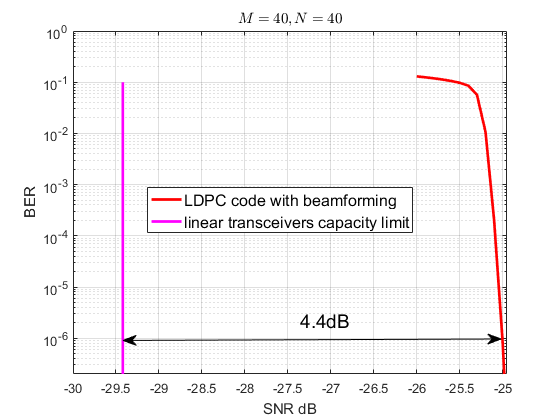}
    \caption{Similar to Fig.~\ref{fig:M_8_N_8_beamforming_LDPC}, except with $M=N=40$.  The gap between the capacity of linear transceivers and the beamformed one-bit transceivers is only 4.4 dB.}
    \label{fig:M_40_N_40_beamforming_LDPC}
\end{figure}

\section*{Acknowledgment}
The authors are grateful for the support of the National Science Foundation, grants ECCS-1731056, ECCS-1509188, and CCF-1403458.

\bibliographystyle{IEEEtran}
\bibliography{IEEEabrv,ITA_2018_refs}

\end{document}